\documentclass[12pt,a4paper,oneside,lefteqn]{article}
\pdfoutput=1

\usepackage[margin=1.5in]{geometry}

\usepackage[T1]{fontenc}
\usepackage[utf8]{inputenc}
\usepackage[american]{babel}

\usepackage[section]{placeins}

\usepackage{hyperref}
\usepackage[pdftex]{graphicx}
\usepackage{xspace}
\usepackage{textcomp}
\usepackage{framed}

\usepackage{verbatim}

\usepackage{amsmath}
\usepackage{amssymb}
\usepackage{amsthm}
\usepackage{txfonts}
\usepackage{mathtools}

\newcommand{\Nat}{\mathbb{N}}

\newcommand{\meet}{\wedge}
\newcommand{\join}{\vee}

\newcommand{\bigmeet}{\bigwedge}
\newcommand{\bigjoin}{\bigvee}

\newcommand{\lcm}{\operatorname{\mathrm{lcm}}}

\newcommand{\upsort}[1]{{#1}^{\uparrow}}

\newcommand{\weakupsort}[1]{{#1}^\vartriangle}

\newcommand{\pascalupsort}[3]{\weakupsort{#1}\,\tbinom{#2}{#3}}

\newcommand{\eqdef}{\coloneqq}

\theoremstyle{definition}
\newtheorem{definition}{Definition}

\theoremstyle{theorem}
\newtheorem{lemma}[definition]{Lemma}

\newtheorem{proposition}[definition]{Proposition}

\title{Recursive Sorting in Lattices}
\author{Jens Gerlach\\
        {\small Fraunhofer FOKUS, Berlin, Germany}
       }
\date{}

\begin{document}



\maketitle

\begin{abstract}
The direct application of the definition of sorting in
lattices~\cite{arxiv-sorting-in-lattices} is impractical
because it leads to an algorithm with exponential complexity.
In this paper we present for \emph{distributive} lattices a recursive formulation to compute
the sort of a sequence.
This alternative formulation is inspired by the identity 
$\binom{n}{k} = \binom{n-1}{k-1} + \binom{n-1}{k}$ that underlies Pascal's triangle.
It provides quadratic complexity and is in fact a generalization of \emph{insertion sort} for lattices.
\end{abstract}

\section{Background}

If someone asked whether there is for the numbers $x$ and $y$ and the exponent~$n$
a general relationship between the value
$(x+y)^n$ and the powers~$x^n$ and~$y^n$, then the (obvious) answer is that
this relationship is captured by the Binomial Theorem
\begin{align*}
    (x+y)^n &= \sum_{k = 0}^n \binom{n}{k}x^{n-k} y^{k}
\end{align*}
which also shows that other powers of $x$ and $y$ are involved.

If, on the other hand, $(x_1,\ldots,x_n)$ is a sequence in a \emph{totally ordered} set $(X,\leq)$
and someone asked whether there is a general relationship between the elements
of~$x$ and the elements of its nondecreasingly sorted 
counterpart~$\upsort{x} = \left(\upsort{x}_1,\ldots,\upsort{x}_n\right)$,
then one could provide an easy answer for the first and last elements of~$\upsort{x}$.
In fact, we know that $\upsort{x}_1$ is the least element of $\{x_1,\ldots,x_n\}$
\begin{align*}
  \upsort{x}_1 &= x_1 \meet \ldots \meet x_n  = \bigmeet_{k = 1}^n x_k,
\intertext{whereas $\upsort{x}_n$ is the greatest element of $\{x_1,\ldots,x_n\}$}
  \upsort{x}_n &= x_1 \join \ldots \join x_n = \bigjoin_{k = 1}^n x_k.
\end{align*}

The relationship between an arbitrary element $\upsort{x}_k$ and the elements
of~$x$ reads
\begin{align}
\label{eq:sorting-equivalence}
  \upsort{x}_k  &=  \bigmeet_{I \in \Nat\tbinom{n}{k}} \bigjoin_{i \in I} x_i
\end{align}
and has been proven in a previous work of the author~\cite[Proposition~2.2]{arxiv-sorting-in-lattices}.
Here~$\Nat\tbinom{n}{k}$ is the set of subsets of~$[1,n]$ that contain exactly $k$ elements.
Note that $\Nat\tbinom{n}{k}$ consists of $\tbinom{n}{k}$ elements.

Equation~\eqref{eq:sorting-equivalence} is not just a compact formula for 
computing~$\upsort{x}$.
It also provides a way to generalize the concept of sorting 
\emph{beyond} totally ordered sets.
In fact, if $(X\leq)$ is a partically ordered set, that is also a lattice $(X,\meet,\join)$,
then for each finite subset~$A$ of~$X$
both the infimum and supremum of~$A$ exist (denoted by $\bigmeet A$ and $\bigjoin A$, respectively).
Thus, the right hand side of Equation~\eqref{eq:sorting-equivalence} is well-formed
in a lattice.
Therefore we can define (as in~\cite[Definition~3.1]{arxiv-sorting-in-lattices})
for $x = (x_1,\ldots,x_n)$ a new 
sequence~$\weakupsort{x} = \left(\weakupsort{x}_1,\ldots,\weakupsort{x}_n\right)$ by
\begin{align}
\label{eq:weakupsort}
   \weakupsort{x}_k &\eqdef \bigmeet_{I \in \Nat\tbinom{n}{k}} \bigjoin_{i \in I} x_i.
\end{align}
We refer to $\weakupsort{x}$ as \emph{$x$ sorted with respect to the lattice $(X,\meet,\join)$}.
It can be shown that this definition of sorting in lattices
maintains many properties that are familiar from sorting in totally ordered sets.
For example, the sequence $\left(\weakupsort{x}_1,\ldots,\weakupsort{x}_n\right)$ is 
nondecreasing~\cite[Lemma~2.3]{arxiv-sorting-in-lattices} and 
the mapping $x \mapsto \weakupsort{x}$ is idempotent~\cite[Lemma~3.6]{arxiv-sorting-in-lattices}.

Note that we reserve the notation $\upsort{x}$ in order to refer
to \emph{sorting in totally ordered sets} whereas we use 
the notation~$\weakupsort{x}$ to refer to \emph{sorting in lattices}.

\section{The need for a more efficient formula}

Definition~\eqref{eq:weakupsort} is nice and succinct, but it is also are quite
impractical to use in computations.
While conducting some experiments with Equation~\eqref{eq:weakupsort} in the
lattice~$(\Nat,\gcd,\lcm)$ it became obvious that only for very short sequences~$x$
the sequence~$\weakupsort{x}$ can be computed in a reasonable time.

Table~\ref{tbl:performance-exp} shows simple
performance measurements (conducted on a notebook computer) for 
computing $\weakupsort{(1,\ldots,n)}$ in~$(\Nat,\gcd,\lcm)$.
The reason for this dramatic slowdown is of course the exponential complexity
inherent in Equation~\eqref{eq:weakupsort}: In order to
compute~$\weakupsort{x}$ from~$x$ it is necessary to consider
all~$2^n-1$ nonempty subsets of~$[1,n]$.

\clearpage

\begin{table}[hbt]
\begin{center}
\begin{tabular}{|l|c|c|c|c|c|c|c|}
\hline
  \textbf{sequence length}  & 20  &  21  & 22  &  23  & 24   & 25   & 26   \\
\hline
  \textbf{time in $s$} & 0.6 &  1.3 & 2.7 &  5.8 & 11.8 & 25.5 & 51.6 \\
\hline
\end{tabular}
\end{center}
\caption{\label{tbl:performance-exp} Wall-clock time for computing~$\weakupsort{(1,\ldots,n)}$
    according to Equation~\eqref{eq:weakupsort}}
\end{table}

In order to address this problem, we prove in Proposition~\ref{recursive-sorting}
the recursive Identity~\eqref{eq:pascal-identity} for the case of
\emph{bounded distributive lattices}.
This identity is closely related to the well-known fact that the binomial coefficient
\begin{align*}
   \binom{n}{k} &= \frac{n!}{k! \cdot (n-k)!}
\intertext{can be efficiently computed through the recursion underlying Pascal's triangle}
   \binom{n}{k} &= \binom{n-1}{k-1} + \binom{n-1}{k}.
\end{align*}

Furthermore, we prove in Proposition~\ref{recursive-sorting-converse} that
a lattice, in which Identity~\eqref{eq:pascal-identity} holds, is necessarily distributive.

\section{Recursive sorting in lattices}

For the remainder of this paper we assume that $(X,\meet,\join, \bot, \top)$ is a 
\emph{bounded} lattice.
Here~$\bot$ is the least element of $X$ and the neutral element of join
\begin{align}
\label{eq:bot}
  x &= \bot \join x = x \join \bot  && \forall x \in X.
\intertext{At the same time, $\top$ is the greatest element of $X$ and the neutral element of meet}
\label{eq:top}
  x &= \top \meet x = x \meet \top  && \forall x \in X.
\end{align}

We are now introducing a notation that allows us to
concisely refer to individual elements of both $\weakupsort{(x_1,\ldots,x_n)}$
and $\weakupsort{(x_1,\ldots,x_{n-1})}$.
Here again, it is convenient to employ the symbol for the binomial coefficient~$\tbinom{n}{k}$
in the context of sorting in lattices.

For a sequence $x$ of length~$n$ we define for $0 \leq m \leq n$
\begin{align}
\label{eq:pascalsort}
    \pascalupsort{x}{m}{k} 
       &\eqdef\begin{cases}
                \bot & \qquad k = 0 \\
                \weakupsort{(x_1,\ldots,x_m)}(k)  & \qquad  k \in [1,m]\\
                \top & \qquad k = m+1
               \end{cases}
\end{align}

We know from the definition of $\weakupsort{x}$ in Equation~\eqref{eq:weakupsort} that
\begin{align}
\nonumber
    \weakupsort{(x_1,\ldots,x_m)}(k) 
        &= \bigmeet_{I \in \Nat\binom{m}{k}} \bigjoin_{i \in I} x_i 
\intertext{holds for $1 \leq k \leq m$. We therefore have for $1 \leq k \leq m$}
\label{eq:pascalsort-general}
   \pascalupsort{x}{m}{k} 
        &= \bigmeet_{I \in \Nat\binom{m}{k}} \bigjoin_{i \in I} x_i.
\intertext{In particular, the identity}
\label{eq:pascalsort-special}
   \pascalupsort{x}{n}{k} &= \weakupsort{x}_k 
\end{align}
holds for $1 \leq k \leq n$.

The main result of this paper is Proposition~\ref{recursive-sorting}, which
states in Identity~\eqref{eq:pascal-identity},
how the $k\text{th}$ element of $\weakupsort{(x_1,\ldots,x_n)}$ can be computed
from~$\weakupsort{(x_1,\ldots,x_{n-1})}$ and~$x_n$
by simply applying one \emph{join} and one \emph{meet}.

The proof of Proposition~\ref{recursive-sorting} relies on the fact that the lattice
under consideration is both \emph{bounded} and \emph{distributive}.
The boundedness of~$X$ is, in contrast to its distributivity, no real restriction because
every lattice can be turned into a bounded lattice by
adjoining a smallest and a greatest element~\cite[p.~7]{roman2008lattices}.

\begin{proposition}
\label{recursive-sorting}
If $(X,\meet,\join, \bot, \top)$ is a bounded distributive lattice
and if~$x$ is a sequence of length~$n$, then 
\begin{align}
\label{eq:pascal-identity}
   \pascalupsort{x}{n}{k} &=
        \pascalupsort{x}{n-1}{k} \meet \biggl(\pascalupsort{x}{n-1}{k-1} \join x_{n}\biggr)
\end{align}
holds for $1 \leq k \leq n$.
\end{proposition}

\begin{proof}
We first consider the ``corner cases'' $k = 1$ and $k = n$.

For $k = 1$, we have
\begin{align*}
   \pascalupsort{x}{n}{1} 
         &= \bigmeet_{i = 1}^{n} x_i 
              && \text{by Identity~\eqref{eq:pascalsort-general}}\\
         &= \left(\bigmeet_{i = 1}^{n-1} x_i\right) \meet x_{n} 
             && \text{by associativity}\\
         &= \pascalupsort{x}{n-1}{1} \meet x_{n} 
              && \text{by Identity~\eqref{eq:pascalsort-general}}\\
         &= \pascalupsort{x}{n-1}{1} \meet \biggl(\bot \join x_{n}\biggr) 
              && \text{by Identity~\eqref{eq:bot}}\\
         &= \pascalupsort{x}{n-1}{1} \meet \biggl(\pascalupsort{x}{n-1}{0} \join x_{n}\biggr) 
              && \text{by Identity~\eqref{eq:pascalsort}.}
\end{align*}
 
For $k = n$, we have
\begin{align*}
  \pascalupsort{x}{n}{n} 
         &= \bigjoin_{i = 1}^{n} x_i 
              && \text{by Identity~\eqref{eq:pascalsort-general}}\\
         &= \left(\bigjoin_{i = 1}^{n-1} x_i\right) \join x_{n} 
             && \text{by associativity}\\
         &= \pascalupsort{x}{n-1}{n-1} \join x_{n} 
              && \text{by Identity~\eqref{eq:pascalsort-general}}\\
         &= \biggl(\pascalupsort{x}{n-1}{n-1} \join x_{n}\biggr) \meet \top
              && \text{by Identity~\eqref{eq:top}}\\
         &=  \biggl(\pascalupsort{x}{n-1}{n-1} \join x_{n}\biggr) \meet \pascalupsort{x}{n-1}{n}
              && \text{by Identity~\eqref{eq:pascalsort}}  \\
         &= \pascalupsort{x}{n-1}{n} \meet \biggl(\pascalupsort{x}{n-1}{n-1} \join x_{n}\biggr) 
              && \text{by commutativity.}
\end{align*}

In the general case of $1 < k < n$,
we first remark that if $A$ is a subset of~$[1,n]$ that consists of $k$ elements, then
there are two cases possible:

\begin{enumerate}
\item If $n$ does not belong to $A$, then $A$ is a subset of~$\Nat\tbinom{n-1}{k}$.
\item If $n$ is an element of $A$, then the set $B \eqdef A \setminus \{n\}$
      belongs to~$\Nat\tbinom{n-1}{k-1}$.
\end{enumerate}

In other words, $\Nat\tbinom{n}{k}$ can be represented as the
following (disjoint) union
\begin{align}
\label{eq:pascalsort-union}
 \Nat\tbinom{n}{k} &=  \Nat\tbinom{n-1}{k} \cup \left\{ B \cup \{n\} \bigm| B \in \Nat\tbinom{n-1}{k-1} \right\}.
\end{align}

We conclude

\begin{xalignat*}{5}
   \pascalupsort{x}{n}{k}
      &= \bigmeet_{I \in \Nat\binom{n}{k}} \bigjoin_{i \in I} x_i
           &&
           &&
           && \text{by Identity \eqref{eq:pascalsort-general}} \\
\intertext{}
      &= \bigmeet_{I \in \Nat\binom{n-1}{k}} \bigjoin_{i \in I} x_i
           &&\meet 
           &&\quad\bigmeet_{I \in \Nat\binom{n-1}{k-1}} \bigjoin_{i \in I \cup \{n\}} x_i
           && \text{by Identity \eqref{eq:pascalsort-union}} \\
\intertext{}
      &= \quad\pascalupsort{x}{n-1}{k}
           &&\meet 
           &&\quad\bigmeet_{I \in \Nat\binom{n-1}{k-1}} \bigjoin_{i \in I \cup \{n\}} x_i
           && \text{by Identity \eqref{eq:pascalsort-general}} \\
\intertext{}
      &= \quad\pascalupsort{x}{n-1}{k}
           &&\meet 
           &&\bigmeet_{I \in \Nat\binom{n-1}{k-1}} \left(\bigjoin_{i \in I} x_i \join x_n \right)
           && \text{by associativity}\\
\intertext{}
      &= \quad\pascalupsort{x}{n-1}{k}
           &&\meet 
           &&\left(
               \left(\bigmeet_{I \in \Nat\binom{n-1}{k-1}} \bigjoin_{i \in I} x_i \right) \join x_n
             \right)
           && \text{by distributivity}\\
\intertext{}
      &= \quad\pascalupsort{x}{n-1}{k}
           &&\meet 
           &&\quad\biggl(\pascalupsort{x}{n-1}{k-1} \join x_n\biggr)
           && \text{by Identity \eqref{eq:pascalsort-general}} 
\end{xalignat*}
which concludes the proof.
\end{proof}

The following Proposition~\ref{recursive-sorting-converse} states that
the converse of Proposition~\ref{recursive-sorting} also holds.

\begin{proposition}
\label{recursive-sorting-converse}
Let $(X,\meet,\join,\bot,\top)$ be a bounded lattice which is \emph{not} distributive.
Then there exists a sequence $x = (x_1,x_2,x_3)$ in $X$ such that Identity~\eqref{eq:pascal-identity}
is not satisfied.
\end{proposition}

\begin{proof}
According to a standard result on distributive lattices~\cite[Theorem~4.7]{roman2008lattices},
a lattice is \emph{not} distributive, if and only if it contains a sublattice
which is isomorphic to either~$N_5$ or~$M_3$ (see Figure~\ref{fig:n5-m3}).

\begin{figure}[hbt]
\begin{center}
\includegraphics[width=0.85\linewidth]{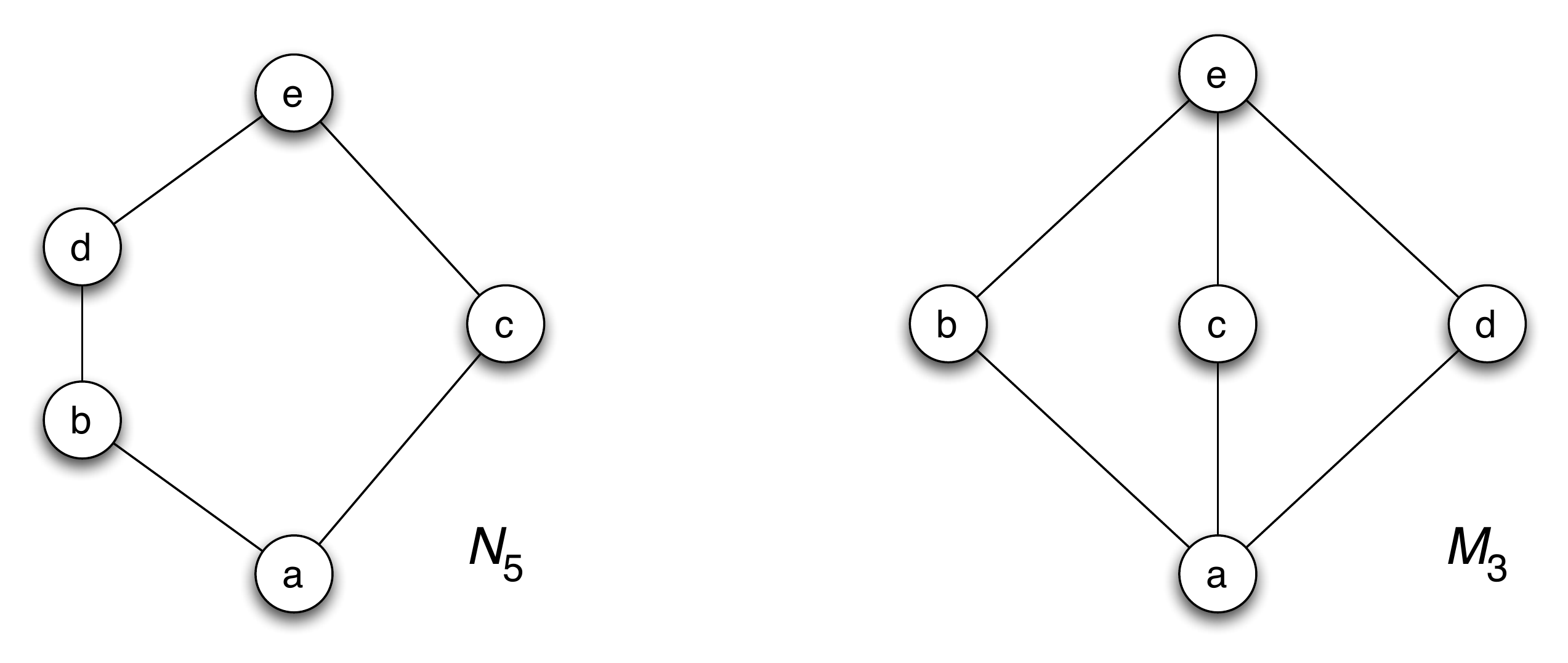}
\end{center}
\caption{\label{fig:n5-m3} The non-distributive lattices $N_5$ and $M_3$}
\end{figure}

From Identity~\eqref{eq:weakupsort} (see also~\cite[Identity~9]{arxiv-sorting-in-lattices})
follows for the elements of 
$\weakupsort{x} = (\weakupsort{x}_1, \weakupsort{x}_2, \weakupsort{x}_3)$
\begin{subequations}
\label{eq:weakupsort3}
\begin{align}
  \weakupsort{x}_1 &= x_1 \meet x_2 \meet x_3 \\
  \weakupsort{x}_2 &= (x_1 \join x_2) \meet (x_1 \join x_3) \meet (x_2 \join x_3) \\
  \weakupsort{x}_3 &= x_1 \join x_2 \join x_3.
\end{align}
\end{subequations}

If $X$ contains the sublattice $N_5$, then we consider the sequence $x = (c,d,b)$ and 
its subsequence $(c,d)$.
From Identity~\eqref{eq:weakupsort3} then follows
\begin{align*}
    \weakupsort{(c,d,b)} &= (a,d,e) \\
    \weakupsort{(c,d)} &= (a,e).
\end{align*}
Thus, we have
\begin{xalignat*}{3}
    \pascalupsort{x}{3}{2} &= d &
    \pascalupsort{x}{2}{2} &= e &
    \pascalupsort{x}{2}{1} &= a.
\end{xalignat*}
However, applying Identity~\eqref{eq:pascal-identity} we obtain
\begin{align*}
    \pascalupsort{x}{3}{2} 
          &= \pascalupsort{x}{2}{2} \meet \left(\pascalupsort{x}{2}{1} \join x_3\right) \\
          &= e \meet \left(a \join b\right)\\
          & = e \meet b \\
          &= b && \text{instead of $d$.}
\end{align*}

If $X$ contains the sublattice $M_3$, then we consider the sequence $x = (b,c,d)$ and 
its subsequence $(b,c)$.
From Identity~\eqref{eq:weakupsort3} then follows
\begin{align*}
    \weakupsort{(b,c,d)} &= (a,e,e) \\
    \weakupsort{(b,c)} &= (a,e).
\end{align*}
We therefore have
\begin{xalignat*}{3}
    \pascalupsort{x}{3}{2} &= e &
    \pascalupsort{x}{2}{2} &= e &
    \pascalupsort{x}{2}{1} &= a.
\end{xalignat*}
Again, applying Identity~\eqref{eq:pascal-identity} we obtain
\begin{align*}
    \pascalupsort{x}{3}{2} 
          &= \pascalupsort{x}{2}{2} \meet \left(\pascalupsort{x}{2}{1} \join x_3\right) \\
          &= e \meet \left(a \join d\right) \\
          &= e \meet d \\
          &= d && \text{instead of $e$.}
\end{align*}
\end{proof}

Using Identity~\eqref{eq:pascal-identity}, we can prove
following Lemma~\ref{last-is-largest}, which generalizes a known fact known from sorting in 
a total order: If one knows that $x_n$ is greater or equal that the
preceding elements $x_1,\ldots,x_{n-1}$ then
sorting the sequence $(x_1,\ldots,x_n)$ can be accomplished 
by sorting $(x_1,\ldots,x_{n-1})$ and simply appending~$x_n$.

\begin{lemma}
\label{last-is-largest}
Let $(X,\meet,\join, \bot, \top)$ be a bounded distributive lattice
and~$x$ be a sequence of length~$n$.
If the condition $x_i \leq x_n$ holds for $1 \leq i \leq n-1$, then
\begin{align*}
   \pascalupsort{x}{n}{n} &= x_n
\intertext{and}
   \pascalupsort{x}{n}{i} &= \pascalupsort{x}{n-1}{i} 
\end{align*}
holds for $1 \leq i \leq n-1$.
\end{lemma}

\begin{proof}
The first equation follows directly from the fact that $\weakupsort{x}_n$
is the supremum of the values $x_1,\ldots,x_n$.

Regarding the second equation, we known 
from~\cite[Lemma~3.3]{arxiv-sorting-in-lattices} that
if for $1 \leq i \leq n-1$ the inequality $x_i \leq x_n$ holds, then 
\begin{align*}
   \pascalupsort{x}{n-1}{i} &\leq x_n.
\intertext{This inequality is also valid for $i = 0$ because}
\pascalupsort{x}{n-1}{0} = \bot\\
\intertext{holds by Identity~\eqref{eq:pascalsort}.
 From general properties of meet and join then follows that}
   \pascalupsort{x}{n-1}{i} \join x_n &= x_n \\
   \pascalupsort{x}{n-1}{i} \meet x_n &= \pascalupsort{x}{n-1}{i}
\intertext{holds for $0 \leq i \leq n-1$.
We can therefore simplify Identity~\eqref{eq:pascal-identity}}
   \pascalupsort{x}{n}{i} &= \pascalupsort{x}{n-1}{i} \meet \biggl(\pascalupsort{x}{n-1}{i-1} \join x_{n}\biggr)
\intertext{first to}
     &= \pascalupsort{x}{n-1}{i} \meet x_{n}
\intertext{and finally to}
     &= \pascalupsort{x}{n-1}{i}.
\end{align*}
\end{proof}

\section{Insertion sort in lattices}

Equation~\eqref{eq:pascal-sorting-symbolic} symbolically represents
Identity~\eqref{eq:pascal-identity}.
Whenever an arrow~$\searrow$ and and arrow~$\swarrow$ meet, the
values are combinedby a \emph{meet}.
In the case of an arrow~$\searrow$, however, first the value at the origin of the arrow
is combined with the sequence value~$x_n$ through a \emph{join}.
%
\begin{align}
  \label{eq:pascal-sorting-symbolic}
  \begin{array}{cccccc}
      & \pascalupsort{x}{n-1}{k-1}  & & & & \pascalupsort{x}{n-1}{k} \\
      & & \searrow & &  \swarrow & \\
      x_n      & \longrightarrow & & \pascalupsort{x}{n}{k} & &
  \end{array}
\end{align}

\clearpage

Figure~\ref{fig:pascal-sort} integrates several instance of 
Equation~\eqref{eq:pascal-sorting-symbolic} in order to graphically represent 
Identity~\eqref{eq:pascal-identity} and to emphasize its close relationship
to Pascal's triangle.

\begin{figure}[hbt]
\begin{center}
\includegraphics[width=0.95\linewidth]{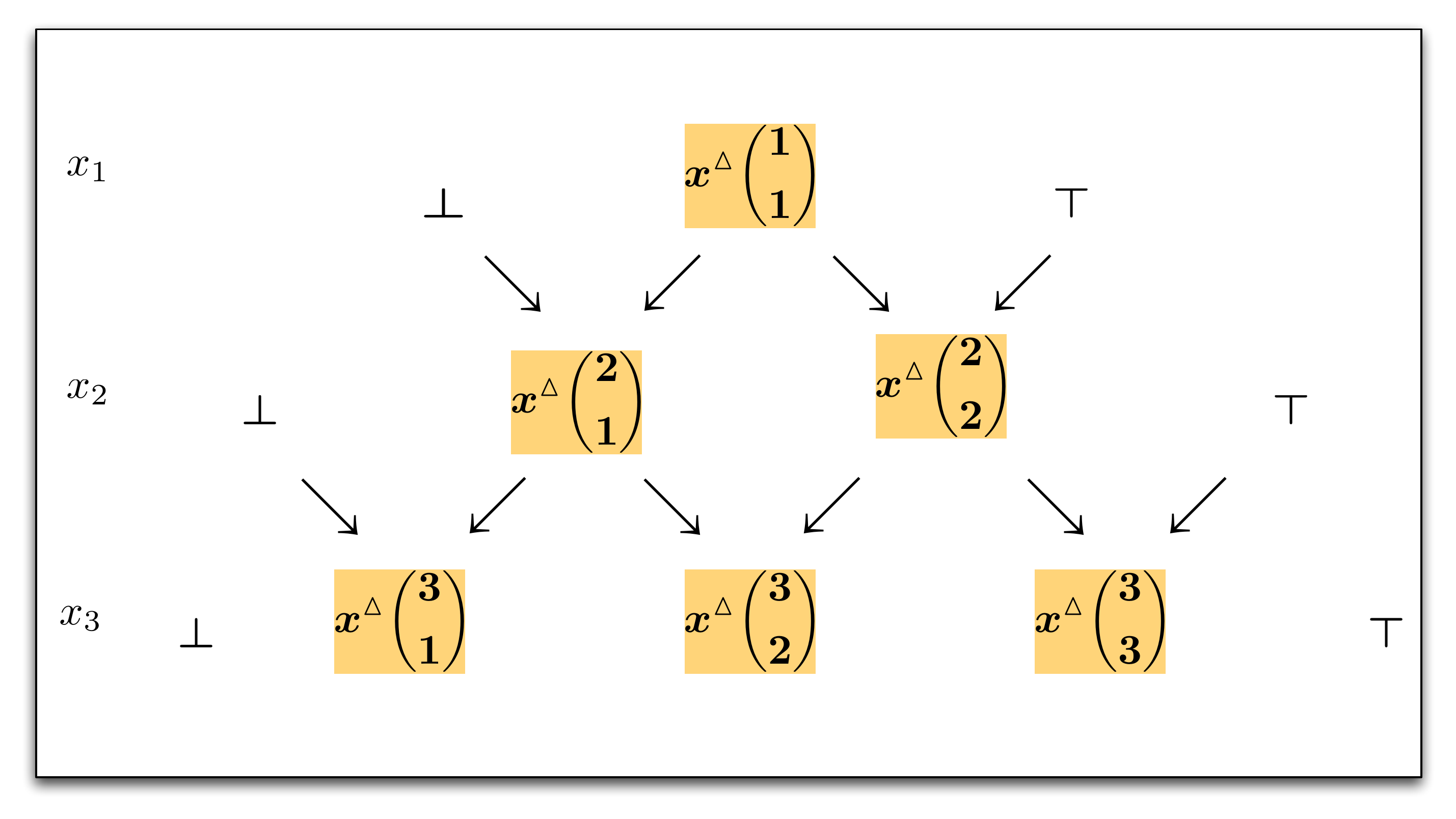}
\end{center}
\caption{\label{fig:pascal-sort} A graphical representation of Identity~\eqref{eq:pascal-identity}}
\end{figure}

Figure~\ref{fig:insertion-sort} outlines an algorithm, which
is based on Identity~\eqref{eq:pascal-identity},
and that starting from $\weakupsort{(x_1)} = (x_1)$ successively 
computes~$\weakupsort{(x_1,\ldots,x_i)}$ for $2 \leq i \leq n$.

\begin{figure}[hbt]
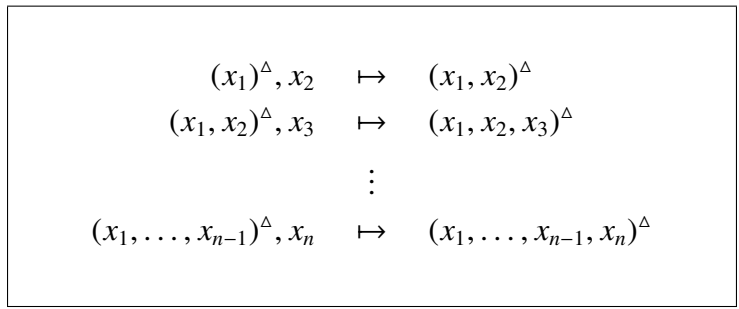

\begin{center}
\fbox{\begin{minipage}{0.7\textwidth}
\begin{align*}
    \weakupsort{(x_1)}, x_2 &\quad\mapsto\quad \weakupsort{(x_1,x_2)} \\
    \weakupsort{(x_1,x_2)}, x_3 &\quad\mapsto\quad \weakupsort{(x_1,x_2,x_3)} \\
                                &\quad\vdotswithin{\mapsto} \\
    \weakupsort{(x_1,\ldots,x_{n-1})}, x_n &\quad\mapsto\quad \weakupsort{(x_1,\ldots,x_{n-1},x_n)} \\
\end{align*}
\end{minipage}}
\end{center}
\caption{\label{fig:insertion-sort} A simple algorithm based on Identity~\eqref{eq:pascal-identity}}
\end{figure}

From Identity~\eqref{eq:pascal-identity} follows that in step~$i$
exactly~$i$ joins and~$i$ meets must be performed.
Thus, altogether there are
\begin{align*}
    \sum_{i = 2}^{n} 2*i &= n(n+1) - 2
\end{align*}
applications of join and meet.
In other words, such an implementation has quadratic complexity.
The algorithm in Figure~\eqref{fig:insertion-sort} can be considered
as \emph{insertion sort}~\cite[\S~5.2.1]{TAOCP3} for lattices
because one element at a time is added to an already ``sorted'' sequence.

Table~\ref{tbl:performance-quad} shows the results of performance measurements in 
the bounded and distributive lattice~$(\Nat,\gcd,\lcm, 1, 0)$.
Here, we are using an implementation that is based on the algorithm in
Figure~\eqref{fig:insertion-sort}.

\begin{table}[hbt]
\begin{center}
\begin{tabular}{|l|c|c|c|c|}
\hline
  \textbf{sequence length}  & 100  &  1000  & 10000 & 100000  \\
\hline
  \textbf{time in $s$} & 0 &  0 & 3.4 &  420  \\
\hline
\end{tabular}
\end{center}
\caption{\label{tbl:performance-quad} Wall-clock time for computing~$\weakupsort{(1,\ldots,n)}$
   according to Equation~\eqref{eq:pascal-identity}}
\end{table}

These results show that sorting in lattices can now be applied to much
larger sequences than those shown in Table~\eqref{tbl:performance-exp}
before the limitations of an algorithm with quadratic complexity become noticeable.

\section{Conclusions}

The main results of this paper are
Proposition~\ref{recursive-sorting} that proves Identity~\eqref{eq:pascal-identity}
for bounded distributive lattices and
Proposition~\ref{recursive-sorting-converse} that shows the necessity of the distributivity
for Identity~\eqref{eq:pascal-identity} to hold.

The remarkable points of Identity~\eqref{eq:pascal-identity} are that it
exhibits a strong analogy between sorting and Pascal's triangle,
allows to sort in lattices with quadratic complexity, and
is in fact a generalization of \emph{insertion sort} for lattices.

\section{Acknowledgment}

I am very grateful for the many corrections and valuable suggestions of my colleagues
Jochen Burghardt and Hans Werner Pohl.

In particular, Jochen Burghardt's suggestion to investigate whether the
distributivity in Proposition~\ref{recursive-sorting} is really necessary led to 
Proposition~\ref{recursive-sorting-converse}.
Hans Werner Pohl pointed out the analogy of the algorithm in
Figure~\ref{fig:insertion-sort} to insertion sort.


\end{document}